\documentclass[runningheads]{llncs}

\usepackage[T1]{fontenc}

\usepackage{amsmath} 
\usepackage{graphicx}  
\usepackage{hyperref}  
\usepackage{color}

\usepackage{doi}

\usepackage{subcaption} 
\usepackage{thmtools,thm-restate} 
\usepackage{cite}

\usepackage{enumitem}
\setlist[enumerate,1]{label=(\roman*)}
\usepackage{amsmath, amssymb, dsfont}
\usepackage{bbm}
\usepackage{standalone}
\usepackage{mathtools}
\newcommand{\N}{\mathbb{N}}

\newcommand{\Q}{\mathbb{Q}}
\newcommand{\R}{\mathbb{R}}
\DeclareMathOperator*{\pp}{\mathcal P}
\DeclareMathOperator*{\s}{\mathcal S}

\DeclareMathOperator{\sign}{sgn}

\DeclareMathOperator{\Eq}{eq}

\newcommand{\gradftop}[1]{\vec{\nabla\!f(#1)}\!^{\top}}

\usepackage{xcolor}
\definecolor{grayif}{gray}{0.7}

\usepackage{tikz}
\usetikzlibrary{arrows.meta,positioning,calc}

\usepackage[ruled,vlined]{algorithm2e}
\usepackage{setspace}

\SetCommentSty{mycommfont}
\newcommand{\vgap}{\vspace{.2em}}
\SetKw{KwAnd}{and}
\SetKw{KwOr}{or}
\SetKw{KwNot}{not}
\DontPrintSemicolon
\SetKw{KwReturn}{return}
\SetKwFor{For}{for}{\!:}{endfor}
\SetKwFor{While}{while}{\!:}{endw}
\SetKwFor{ForEach}{for each}{\!:}{endfor}
\SetKwIF{If}{ElseIf}{Else}{if}{\!:}{else if}{else}{endif}
\SetKwComment{tcp}{(}{)}

\iftrue
\makeatletter
\renewcommand{\SetKwInOut}[2]{%
  \sbox\algocf@inoutbox{\KwSty{#2}\algocf@typo:}%
  \expandafter\ifx\csname InOutSizeDefined\endcsname\relax
    \newcommand\InOutSizeDefined{}\setlength{\inoutsize}{\wd\algocf@inoutbox}%
    \sbox\algocf@inoutbox{\parbox[t]{\inoutsize}{\KwSty{#2}\algocf@typo:\hfill}~}\setlength{\inoutindent}{\wd\algocf@inoutbox}%
  \else
    \ifdim\wd\algocf@inoutbox>\inoutsize%
    \setlength{\inoutsize}{\wd\algocf@inoutbox}%
    \sbox\algocf@inoutbox{\parbox[t]{\inoutsize}{\KwSty{#2}\algocf@typo:\hfill}~}\setlength{\inoutindent}{\wd\algocf@inoutbox}%
    \fi%
  \fi
  \algocf@newcommand{#1}[1]{%
    \ifthenelse{\boolean{algocf@inoutnumbered}}{\relax}{\everypar={\relax}}%
    {\let\\\algocf@newinout\hangindent=\inoutindent\hangafter=1\parbox[t]{\inoutsize}{\KwSty{#2}\algocf@typo:\hfill}~##1\par}%
    \algocf@linesnumbered
  }}%
\makeatother
\fi
\SetKwInOut{Input}{input}
\SetKwInOut{Output}{output}

\usepackage{lscape}  
\usepackage{afterpage}

\usepackage{orcidlink}

\renewcommand{\vec}[1]{\ensuremath{\boldsymbol{#1}}}
\renewcommand{\subset}{\subseteq}

\renewcommand{\epsilon}{\varepsilon}
\newcommand{\p}{\ensuremath{\textup{\textsf{P}}}\xspace}
\newcommand{\np}{\ensuremath{\textup{\textsf{NP}}}\xspace}
\newcommand{\fct}[1]{\mathcal{F}_{#1}}

\begin{document}
\title{An unconditional lower bound for the\\ active-set method on the hypercube}

\titlerunning{An unconditional lower bound for the active-set method on the hypercube}

\author{Yann Disser \orcidlink{0000-0002-2085-0454} \\
Nils Mosis \orcidlink{0000-0002-0692-0647}}

\authorrunning{Y. Disser and N. Mosis}

\institute{TU Darmstadt, Germany\\
\email{\{disser,mosis\}@mathematik.tu-darmstadt.de}}

\maketitle          

\begin{abstract}
The existence of a polynomial-time pivot rule for the simplex method is a fundamental open question in optimization. 
While many super-polynomial lower bounds exist for individual or very restricted classes of pivot rules, there currently is little hope for an unconditional lower bound that addresses all pivot rules.
We approach this question by considering the active-set method as a natural generalization of the simplex method to non-linear objectives.
This generalization allows us to prove the first unconditional lower bound \emph{for all pivot rules}.
More precisely, we construct a multivariate polynomial of degree linear in the number of dimensions such that the active-set method started in the origin visits all vertices of the hypercube.
We hope that our framework serves as a starting point for a new angle of approach to understanding the complexity of the simplex method.


\end{abstract}

\section{Introduction}

The simplex method~\cite{dantzig1955generalized} is widely regarded as one of the most natural and important practical algorithms for solving linear programming problems. 
However, its theoretical complexity remains an open question to this day. 
This has to do with the fact that the simplex method’s behavior is highly sensitive to the choice of pivot rule, which determines how ties between improving edges at a vertex are broken and significantly influences the method’s behavior, such as whether it guarantees termination.

The discovery of a polynomial-time pivot rule for the simplex method would have profound implications. 
In particular, it would most likely provide a strongly polynomial algorithm for linear programming, thereby resolving the ninth open problem on Smale’s list of mathematical problems for the 21st century~\cite{smale2000mathematical}. 
While weakly polynomial algorithms for linear optimization are known, namely the ellipsoid and the interior-point method~\cite{kozlov1979polynomial,karmarkar1984new}, and while the latter has recently been adapted to a strongly polynomial algorithm for linear programs with at most two variables per constraint~\cite{allamigeon2022interior,dadush2024strongly}, barrier methods cannot yield strongly polynomial guarantees~\cite{allamigeon2022no}.
For linear programs with zero-one vertices strongly polynomial simplex pivot rules are known~\cite{black2024simplex}.
A general polynomial-time pivot rule for the simplex method would also imply a polynomial bound on the combinatorial diameter of polytopes, thus resolving the long-standing polynomial Hirsch conjecture~\cite{dantzig1963linear,santos2012counterexample}.

Over the years, super-polynomial lower running time bounds were established for various natural pivot rules. 
Bounds for straightforward rules were found early-on by utilizing distorted hypercubes~\cite{avis1978notes,murty1980computational,klee1972good,jeroslow1973simplex,goldfarb1979worst,black2024exponentiallowerboundspivot}.
More recently, linear programs derived from Markov Decision Processes have been used to handle more sophisticated, memory-based rules~\cite{disser2023exponential,avis2017exponential,friedmann2011subexponential_1,friedmann2011subexponential_2,friedmann2011subexponential_3}.
In some cases, the exponential worst-case behavior of the simplex method could be attributed to its mightiness~\cite{disser2018simplex,fearnley2015complexity,adler2014simplex}, which can be seen as a conditional lower bound.
Nevertheless, results from smoothed complexity~\cite{huiberts2023upper,dadush2018friendly,deshpande2005improved,spielman2004smoothed} suggest that the constructions underlying the above bounds are fragile, which indicates that fundamentally different approaches may be necessary to achieve an unconditional lower bound that applies universally across all pivot rules.
While existing hypercube constructions were generalized to small classes of pivot rules~\cite{amenta1999deformed,disser2023unified}, the existence of an unconditional super-polynomial lower bound remains wide open.

We approach this question by considering the active-set method~\cite{fletcher2000practical,nocedal1999numerical} as a natural generalization of the simplex method to non-linear programs.
A case of particular importance for the active-set method are quadratic, concave programs for which, much like in the linear setting, only weakly polynomial algorithms are known~\cite{kozlov1979polynomial,ye1989extension}, and the active-set method is a promising candidate for the first strongly polynomial algorithm.
Similarly to the simplex method, the active-set method is governed by a pivot rule that determines the improving direction in each step in case there is more than one option.
Again, no general running time bounds for all pivot rules are known.

\subsubsection{Our results.}
We prove the first unconditional super-polynomial lower bounds on the running time of the active-set method \emph{for all pivot rules}.
More precisely, we show the following.

\begin{theorem}\label{thm:main}
For all~$n\in\N_{>2}$, there is a multivariate polynomial~$\fct n$ of degree~$n$ such that the active-set method started in~$\vec{0}$ needs~$2^n-1$ iterations to optimize~$\fct n$ over the~$n$-dimensional hypercube~$[0,1]^n$, irrespective of the pivot rule.
\end{theorem}

In particular, our construction forces the active-set method to visit all vertices of the hypercube while traversing along its edges in a simplex-like fashion. 

This result yields a new method of approach for unconditional lower bounds on the running time of the simplex method:
If we could lower the degree of the polynomial to~$1$ by considering more general polytopes, we would obtain a super-polynomial (even exponential) lower bound for the simplex method using any pivot rule.

As a first step in this direction, we can lower the degree of the polynomial, while maintaining a super-polynomial bound, simply by neglecting most of the input dimensions.

\begin{corollary}\label{cor:main2}
For every~$d(n)=\omega(\log n)$ with~$d(n)\leq n$, there are polynomials~$(\fct n)_{n \in \N}$ of degrees~$\mathcal{O}(d(n))$ such that the active-set method started in~$\vec{0}$ needs~$2^{\omega(\log n)}$ iterations to optimize~$\fct n$ over the~$n$-dimensional hypercube~$[0,1]^n$, irrespective of the pivot rule.
\end{corollary}

Importantly, we provide a single construction that addresses all pivot rules simultaneously (there is a unique improving direction in each step).
This means that lowering the degree of the polynomial for general polytopes would constitute progress towards disproving the monotone polynomial Hirsch conjecture~\cite{ziegler2012lectures}.

We observe that it is \np-complete to optimize polynomials of degree~3 over the hypercube (Proposition~\ref{prop: NP complete}).
This already implies a super-polynomial lower bound on the running time of the active-set method, assuming $\p \neq \np$.
The value of our contribution lies in the fact that we provide an \emph{unconditional} lower bound.

In fact, it turns out that the polynomials of Theorem~\ref{thm:main} and Corollary~\ref{cor:main2} induce decomposable unique sink orientations of the hypercube (Proposition~\ref{prop: induces USO}) that can be solved in linear time~\cite{SchurrS04}. In particular, the polynomials in our construction do not give rise to \np-hard optimization problems (unless $\p = \np$).

\subsubsection{Notation.}
We write all vectors in boldface and denote the~$i$-th unit vector by~$\vec{e^i}$.
For~$x,y\in\N\cup\{0\}$, we write~$x\equiv_2 y$ if both~$x$ and~$y$ are either even or odd.
For~$n\in\N$, we write~$[n]\coloneqq\{1,\ldots,n\}$.
We denote the~$i$-th row of a matrix~$A\in\R^{m\times n}$ by~$\vec{A_{i\cdot}}$ and, for an index set~$\mathcal{A}\subseteq[m]$, we write~$A_{\mathcal{A}\cdot}$ to denote the matrix that consists of the rows~$\vec{A_{i\cdot}}$ with~$i\in\mathcal{A}$.

\section{Simplex and active-set method}\label{sec: simplex and active set}
The aim of this section is to illustrate that the active-set method is a natural generalization of the simplex method.
This connection is well-known. 

\begin{proposition}\emph{(Exercise 8.17 in~\cite{fletcher2000practical})} 
The {active-set method}, when applied to linear objectives, is \emph{equivalent} to the simplex method, i.e., the same intermediate solutions are computed during the application of both algorithms.
\end{proposition} 

Consider Algorithms~\ref{alg:Simplex} \& \ref{alg:AS} for a side-by-side comparison of \textsc{Simplex} and \textsc{ActiveSet}.
Both algorithms are formulated in a way that makes their connection evident.
We now describe the behavior of both algorithms in more detail.

\afterpage{%
\begin{landscape}
    \vspace*{\fill}
\begin{figure}[htb]
  \centering
  \begin{minipage}[t]{.47\linewidth}
    \begin{algorithm}[H]
    \begin{minipage}{\dimexpr\textwidth-10pt} 
        \setstretch{1.2} 
     
	    \label{alg:Simplex}\caption{\textsc{Simplex}($\vec c$, $A$, $\vec b$, $\vec{x}$)}
	    \textbf{input: }$\vec c\in\Q^n$, vertex $\vec{x}$ of polytope~$\mathcal P\coloneqq\{\vec y\colon A\vec y\leq \vec b\}$,\;
	    \Indp system~$A\vec y\leq\vec b$ non-degenerate\;
	    \Indm \textbf{output: }optimum solution~$\vec x$ of~$\max\{\vec c^{\top}\vec y\colon\vec y\in\mathcal{P}\}$\;
	    \vgap
	    \hrule  
	    \vgap
	    
	    $\mathcal{B}\gets\Eq(\vec x)$\tcp*{basis}
        \SetInd{4pt}{4pt} 
	    
	    \While{$\mathcal{D}\coloneqq\{\vec{d}\colon{A_{\Eq(\vec x)\cdot}}\vec d\leq \vec 0$,~$\vec{c}^{\top}\vec{d}> 0\}\neq\emptyset$}
	    {          
	     \vspace{2pt} 
		 take $\vec{d}\in\mathcal{D}$ with~$|\{i\in\mathcal{B}\colon \vec{A_{i\cdot}}\vec{ d}=0\}|=n-1$\; 		 
		 
		 \textcolor{grayif}{\If(\tcp*[f]{always true}){$A_{\mathcal{B}\cdot}\vec{d}\neq\vec 0$}{
		 \textcolor{black}{$\mathcal{B}\gets\mathcal{B}\setminus\{\ell\}$ for the unique~$\ell\in\mathcal{B}$ with~$\vec{A_{\ell\cdot}}\vec d<0$}\;
		}}
		
	    \textcolor{grayif}{\If(\tcp*[f]{always true}){$A_{\mathcal{B}\cdot}\vec{d}=\vec 0$}{
	     \textcolor{black}{${\mu}\gets\inf\{\mu\geq 0\colon \vec{x}+\mu\vec{d}\notin\mathcal{P}\}$}\;
	     \textcolor{black}{$\vec{x}\gets\vec{x}+\mu\vec{d}$}\;
	     \textcolor{grayif}{\If(\tcp*[f]{always true}){$\vec{c}^{\top}\vec d>0$}{
	    \textcolor{black}{$\mathcal{B} \gets \mathcal{B}\cup\{j\}$ for the unique~$j\in\Eq(\vec x)\setminus\mathcal{B}$}\;
	    }}}}
	    }
    \end{minipage}
	\end{algorithm}
  \end{minipage}%
  \hfill
  \begin{minipage}[t]{.47\linewidth}
    \begin{algorithm}[H]
    \begin{minipage}{\dimexpr\textwidth-10pt}
        \setstretch{1.2} 
     
	    \label{alg:AS}\caption{\textsc{ActiveSet}($f$, $A$, $\vec b$, $\vec{x}$)}
	    \textbf{input: }$f\colon\R^n\to\R$ continuously differentiable,\;
	    \Indp point~$\vec{x}$ of polytope~$\mathcal P\coloneqq\{\vec y\colon A\vec y\leq \vec b\}$\;
	    \Indm \textbf{output: }critical point~$\vec x$ of~$\max\{f(\vec y)\colon\vec y\in\mathcal{P}\}$\;
	    \vgap
	    \hrule  
	    \vgap
	    
	    $\mathcal{A}\gets\Eq(\vec x)$\tcp*{active set} 
        \SetInd{4pt}{4pt} 
	    
	    \While{$\mathcal{D}\coloneqq\{\vec{d}\colon{A_{\Eq(\vec x)\cdot}}\vec d\leq \vec 0,~\vec{\nabla\!f(\vec{x})}\!^{\top}\vec{d}> 0\}\neq\emptyset$}
	    {
	    \vspace{2pt}
		 take $\vec{ d}\in\mathcal{D}$ maximizing~$|\{i\in\mathcal{A}\colon \vec{A_{i\cdot}}\vec{ d}=0\}|$\; 	     
	     
	     \If{$A_{\mathcal{A}\cdot}\vec{d}\neq\vec 0$}{
	     $\mathcal{A}\gets\mathcal{A}\setminus\{i\}$ for some~$i\in\mathcal{A}$ with~$\vec{A_{i\cdot}}\vec{d}<0$\;	
		 
	     }{}
	     \If{$A_{\mathcal{A}\cdot}\vec{d}=\vec 0$}{	     
	     ${\mu}\gets\inf\{\mu\geq 0\colon \vec{x}+\mu\vec{d}\notin\mathcal{P}\text{ or } \vec{\nabla\!f(x}+\mu \vec{d)}\!^{\top}\vec d\leq 0\}$\;
	     
	     $\vec{x}\gets\vec{x}+\mu\vec{d}$\;
	     
	     \If{$\vec{\nabla\!f(\vec{x})}\!^{\top}\vec{d}> 0$}
	     {
	     $\mathcal{A} \gets \mathcal{A}\cup\{j\}$ for some~$j\in\Eq(\vec{x})\setminus\mathcal{A}$\;
	     }    	 	
	     }    
	    }
    \end{minipage}
	\end{algorithm}
  \end{minipage}
\end{figure}
    \vspace*{\fill}
\end{landscape}
}

~

\noindent\textbf{The simplex method} solves arbitrary linear programs of the form
\begin{equation}\tag{LP}\label{LP}
\begin{aligned}
    \max \quad & \vec{c}^{\top}\vec{x} \\
    \text{s.t.} \quad & A\vec{x} \leq \vec{b},
\end{aligned}
\end{equation}
where~$\vec c\in\Q^n$,~$A\in\Q^{m\times n}$, and~$\vec b\in\Q^m$.
For simplicity, we assume that the feasible region~$\mathcal{P}\coloneqq\{\vec x\colon A\vec x\leq \vec b\}$ is bounded, and that the constraints are non-degenerate, i.e., we assume that every vertex~$\vec x$ of~$\mathcal{P}$ satisfies~$|\Eq(\vec x)|=n$, where we write~$\Eq(\vec x)\coloneqq\{i\colon \vec{A_{i\cdot}}\vec x=b_i\}$.
The simplex method can be understood combinatorially as a traversal of~$\mathcal{P}$ along its edges.

More precisely, it maintains a subset~$\mathcal{B}$, called a \emph{basis}, of~$n$ constraint indices such that the system~$A_{\mathcal{B}\cdot}\vec x=\vec{b}_{\mathcal{B}}$ uniquely determines the current vertex~$\vec x$. 
By non-degeneracy, this is equivalent to~$\mathcal{B}=\Eq(\vec x)$.
As long as~$\vec x$ is not already an optimum, the simplex method moves along edges of~$\mathcal{P}$ to neighboring vertices with better objective function values.
We now describe how our formulation of \textsc{Simplex} reflects this behavior.

A vertex~$\vec x$ of~$\mathcal{P}$ is not optimal if and only if there is a direction~$\vec{d}$ that is feasible and improving, i.e., that satisfies~$A_{\Eq(\vec x)\cdot}\vec d\leq \vec 0$ and~$\vec{c}^{\top}\vec d>0$.
In particular, there is at least one improving edge direction\footnote{An edge direction in vertex~$\vec x$ is a vector of the form~$\vec d\coloneqq \lambda(\vec y-\vec x)$ where~$\vec y$ is a neighboring vertex of~$\vec x$ and~$\lambda>0$.} in each vertex~$\vec x$ that is not optimal.
Due to non-degeneracy, every edge direction~$\vec d$ corresponds to a unique index~$\ell\in\mathcal{B}$ with~$\vec{A_{\ell\cdot}}\vec d<0$ and~${A_{\mathcal{B}\setminus\{\ell\}\cdot}}\vec d=\vec 0$.
In fact, a feasible direction~$\vec d$ is an edge direction in~$\vec x$ if and only if
\[
|\{i\in\mathcal{B}\colon \vec{A_{i\cdot}}\vec{ d}=0\}|=n-1=|\mathcal{B}|-1.
\]
Now, if the current vertex~$\vec x$ is not optimal, \textsc{Simplex} chooses an improving edge direction~$\vec{d}$, deletes the corresponding index~$\ell$ from~$\mathcal{B}$, and moves along~$\vec d$ as far as possible without becoming infeasible.
Note that, by boundedness of~$\mathcal{P}$, this transitions to another vertex~$\vec{y}$ of~$\mathcal{P}$.
The method then adds a new active constraint index~$j\in\Eq(\vec{y})$ to~$\mathcal{B}$ to obtain a basis uniquely identifying the new vertex.
By non-degeneracy, this index is uniquely determined by~$j\in\Eq(\vec{y})\setminus\mathcal{B}$.

The number of iterations of the simplex method depends on the \emph{pivot rule} that determines the improving edge direction~$\vec d$ chosen in each step in case there is more than one.
The question of whether or not there exists a pivot rule guaranteeing a polynomial running time of the simplex method is arguably one of the most famous open problems in (linear) optimization.
We refer to~\cite{dantzig1963linear} for more details on the simplex method.

~

\noindent\textbf{The active-set method} can be applied to non-linear programs with linear constraints of the form
\begin{equation}\tag{NLP}\label{NLP}
\begin{aligned}
    \max \quad & f(\vec x) \\
    \text{s.t.} \quad & A\vec{x} \leq \vec{b},
\end{aligned}
\end{equation}
where~$f\colon\R^n\to\R$ is continously differentiable,~$A\in\Q^{m\times n}$, and~$\vec b\in\Q^m$.
As before, we denote the feasible region by~$\mathcal{P}$, and the set of indices of the constraints that are active in some~$\vec{x}\in\mathcal{P}$ by~$\Eq(\vec x)$.
For simplicity, we again assume that~$\mathcal{P}$ is bounded.

If~$\vec{x^*}$ is a local optimum of~\eqref{NLP}, then there are no feasible improving directions in~$\vec{x^*}$.
That is, for all~$\vec{d}$ with~${A_{\Eq(\vec{x^*})\cdot}}\vec d\leq \vec 0$, we have~$\gradftop{x^*}\vec d\leq 0$.
The active-set method aims to find a \emph{critical point}, i.e., a point satisfying this necessary optimality condition, by determining the set of constraints that are active in such a point.

More precisely, the active-set method maintains a subset~$\mathcal{A}\subseteq\Eq(\vec x)$, called \emph{active set}, of constraints that are active in the current solution~$\vec x\in\mathcal P$.
As long as~$\vec x$ is not already a critical point, the method tries to move along an improving feasible direction~$\vec d$ that respects the active set, i.e., that satisfies~$A_{\mathcal{A}\cdot}\vec d=\vec 0$. 
If there is no such direction, the method removes as few indices as possible from~$\mathcal{A}$ until such a direction~$\vec d$ can be found.
Now it moves as far as possible along~$\vec d$ without becoming infeasible, or until reaching a point~$\vec y$ with~$\vec{\nabla\!f(\vec{y})}\!^{\top}\vec{d}\leq 0$.
If this movement is stopped by hitting the boundary of~$\mathcal{P}$, the method adds a new active constraint to~$\mathcal{A}$.
This procedure is iterated until a critical point for~\eqref{NLP} is found.
We now describe how our formulation of \textsc{ActiveSet} allows this behavior.

A feasible point~$\vec x$ is \emph{not critical} if and only if there is a feasible improving direction~$\vec d$, i.e., a direction~$\vec d$ with~$A_{\Eq(\vec x)\cdot}\vec d\leq \vec 0$ and~$\vec{\nabla\!f(\vec{x})}\!^{\top}\vec{d}> 0$.
In \textsc{ActiveSet}, we take such a direction~$\vec d$ which also maximizes the number of indices~$j\in\mathcal{A}$ with~$\vec{A_{j\cdot}}\vec d=0$, and can then remove indices~$i$ from~$\mathcal{A}$ with~$\vec{A_{i\cdot}}\vec d<0$ (in separate iterations).
This corresponds to deleting a smallest set~$\mathcal{I}\subseteq\mathcal{A}$ from~$\mathcal{A}$ that ensures the existence of a feasible improving direction~$\vec{d}$ with~${A_{\mathcal{A}\setminus\mathcal{I}\cdot}}\vec d=\vec 0$.
Now the method moves along the improving direction~$\vec d$ as far as possible without becoming infeasible or reaching a point where the derivative of the objective in direction~$\vec d$ is non-positive.
If the directional derivative is still positive in the new point, the movement was stopped by hitting the boundary of~$\mathcal{P}$, so the method adds a new active constraint to~$\mathcal{A}$.

Implementations of the active-set method are usually formulated for strictly concave, quadratic objective functions (see e.g.~\cite{fletcher2000practical,nocedal1999numerical}), where they compute Karush-Kuhn-Tucker (KKT) points~\cite{karush2013minima,kuhn1951nonlinear} of equality constrained subproblems until reaching a KKT point, that is, the unique global optimum solution, of~\eqref{NLP}.
More precisely, given a feasible solution~$\vec x$ of~\eqref{NLP} and an active set~$\mathcal{A}\subseteq\{i\colon\vec{A_{i\cdot}}\vec{x}=b_i\}$ of constraint indices, the method computes a KKT point~$\vec{x^*}$ of
\[
\begin{array}{rrclcl}\tag{$\text{NLP}_{\!\mathcal{A}}$}\label{NLPsub}
            \displaystyle \max & \multicolumn{3}{l}{f(\vec y)}\\
            \textrm{s.t.} & A_{\mathcal{A}\cdot}\vec y  & = & \vec b_{\mathcal{A}}. & &
        \end{array}
\]
Then, it moves from~$\vec{x}$ towards~$\vec{x^*}$ as far as possible without becoming infeasible.
If it reaches~$\vec{x^*}$, the method removes a constraint from the active set to enable further progress.
Otherwise, it adds a new constraint to the active set. 

Note that our formulation of \textsc{ActiveSet} allows this behavior for strictly concave, quadratic functions.
If the current solution~$\vec{x}$ is not already the KKT point of~\eqref{NLPsub}, the direction~$\vec d\coloneqq \vec{x^*}-\vec x$ is improving (by concavity of the objective) and satisfies~$A_{\mathcal{A}\cdot}\vec d=\vec 0$, i.e., maximizes~$|\{i\in\mathcal{A}\colon \vec{A_{i\cdot}}\vec{ d}=0\}|$.
Further, we have~$\vec{x^*}=\vec x+\bar\mu\vec d$ for
\[
\bar\mu\coloneqq\inf\{\mu\geq 0\colon\vec{\nabla\!f(x}+\mu \vec{d)}\!^{\top}\vec d\leq 0\}.
\]
If~$\vec{x}$ is the KKT point of~\eqref{NLPsub}, then there is no feasible improving direction~$\vec d$ with~$A_{\mathcal{A}\cdot}\vec d=\vec 0$, so \textsc{ActiveSet} removes constraints from~$\mathcal{A}$ preventing progress.

While concave quadratic functions allow for an efficient implementation, the computations in an iteration of \textsc{ActiveSet} are non-trivial for general functions.
In the following, we assume that \textsc{ActiveSet} has access to an oracle, which finds~$\vec d$ and computes~$\mu$ in each step.

The running time of the algorithm, i.e., the number of iterations, highly depends on the \emph{pivot rule} that determines the direction~$\vec{d}$ and the indices~$i\in\mathcal{A}$ and~$j\in\Eq(\vec{x})\setminus\mathcal{A}$ in each iteration in case there is more than one choice.
It is an open problem whether there is a polynomial time pivot rule for the active-set method for (strictly) concave, quadratic objective functions.

\section{Unconditional super-polynomial bounds for active-set}\label{sec: exponential bound}
In Section~\ref{sec: polynomial}, we define multivariate polynomials~$\fct n\colon\R^n\to\R$ of degree~$n$ for all~$n>2$.
We then see in Section~\ref{sec: move to vertex} that there is exactly one improving edge in every vertex of the boolean hypercube~$[0,1]^n$ which is not optimal; and that the value of~$\fct n$ increases when moving along this edge to the neighboring vertex.
The path~$\pi$ obtained by starting in~$\vec 0$ and iteratively moving along unique improving edges to neighboring vertices visits all~$2^n$ vertices of the hypercube, see Section~\ref{sec: hamiltonian}. 
In Section~\ref{sec: main results}, we obtain our main results by showing that \textsc{ActiveSet} follows the path~$\pi$.
Finally, we make some remarks on the complexity of the optimization problem induced by our polynomials in Section~\ref{sec: complexity}.

\subsection{Polynomials of linear degree}\label{sec: polynomial}
For each~$n\in\N$, we define a multivariate polynomial~$\fct n\colon\R^n\to\R$ by
\[
    \fct n(\vec x)=\sum\limits_{i=1}^{n}\left(2^{i-1}\alpha_{n,i}(\vec x)-\beta_{n,i}(\vec x)\right),
\]
with~$\vec x=(x_1,\ldots,x_n)\!^{\top}\!\!\in\R^n$ and $\alpha_{n,i},\beta_{n,i}\colon\R^n\to\R$ given by~$\alpha_{n,n+1}(\vec x)=0$ and
\begin{align}
\alpha_{n,i}(\vec x)&=x_i+(1-2x_i)\alpha_{n,i+1}(\vec x), &\forall\, i\in[n], \label{def alpha}\\
\beta_{n,i}(\vec x)&=2^i\left(x_i-x_i^2\right)\left(1-x_{i-1}+\sum\nolimits_{j=1}^{i-2}x_j\right), &\forall\, i\in[n],\notag
\end{align}
where we write~$x_0\coloneqq 1$, i.e., we have~$\beta_{n,1}(\vec x)=0$ for all~$\vec x$.
Whenever~$\vec x$ and its dimension are clear from the context, we write~$\alpha_i\coloneqq\alpha_{n,i}(\vec x)$ and~$\beta_i\coloneqq\beta_{n,i}(\vec x)$.
Note that~$\fct n$ is a multivariate polynomial of degree~$n$ -- except for~$\fct 2$ which is of degree~$3$.

We will analyze the behavior of \textsc{ActiveSet} optimizing~$\fct n$ on the hypercube~$[0,1]^n$ starting in~$\vec 0$.
As a preparation, we now compute the maximum of~$\fct n$ on~$[0,1]^n$ as well as the gradients of~$\fct n$ in the vertices of the hypercube.

\begin{restatable}{lemma}{lemmamax}\label{lemma: maximum}
For all~$n\in\N$, the optimum solution of
\begin{align*}
    \max \quad & \fct n(\vec x) \\
    \text{s.t.} \quad & \vec x \in [0,1]^n
\end{align*}
is attained by the unit vector~$\vec{e^n}$.
\end{restatable}
\begin{proof}
	Observe that~$\alpha_i\in[0,1]$ for all~$\vec x\in[0,1]^n$ and all~$i\in[n+1]$ (e.g., by induction on~$i$).
	Further, by equation~\eqref{def alpha}, it is easy to see that we have~$\alpha_i(\vec x)=1$ for all~$i\in[n]$ if and only if~$\vec x=\vec{e^n}$.	
	Since, for all~$i\in[n]$, we have~$\beta_i(\vec x)\geq 0$ for all~$\vec x\in[0,1]^n$ and~$\beta_i(\vec y)=0$ for all~$\vec y\in\{0,1\}^n$, this proves the statement.
\qed
\end{proof}

\begin{restatable}{lemma}{lemmapartial}\label{lemma: partial F}
For all~$n\in\N$,~$k\in[n]$, and~$\vec x\in\{0,1\}^n$, the~$k$-th partial derivative of~$\fct n$ is given by
\begin{align*}
	 \partial_k\fct n(\vec x) = (1-2\alpha_{k+1})\!\sum_{i=1}^k 2^{i-1}\!\prod_{j=i}^{k-1}(1-2x_j)-2^k(1-2x_k)\!\left(1-x_{k-1}+\sum\limits_{i=1}^{k-2}x_i\!\right)\!, 
\end{align*}
where we write~$x_0\coloneqq 1$, i.e., we have~$\partial_1\fct n(\vec x)=1-2\alpha_2$.
\end{restatable}
\begin{proof}
For all~$i \in [n]$, the value of~$\alpha_i$ only depends on~$x_k$ if~$i \leq k$.
Thus,~$\partial_k\alpha_i=0$ for all~$i>k$.
Further, by equation~\eqref{def alpha}, one can easily verify that
\[
	\partial_k\alpha_i = (1-2\alpha_{k+1})\prod_{j=i}^{k-1}(1-2x_j)
\]
holds for all~$i \leq k$, e.g., by induction on~$i$ starting with~$i=k$.

For all~$i \in [n]$, the function~$\beta_i$ only depends on~$x_k$ if~$i \geq k$.
Thus,~$\partial_k\beta_i=0$ for all~$i<k$.
For~$i>k$, we have~$\partial_k\beta_i=0$ since~$(x_i-x_i^2)=0$ for all~$\vec{x}\in\{0,1\}^n$ and all~$i\in[n]$.
We obtain the statement by observing that%
\[
    \partial_k\beta_k=2^k(1-2x_k)\left(1-x_{k-1}+\sum_{i=1}^{k-2}x_i\right). \tag*{\raisebox{-4ex}\qed}
\]
\end{proof}

\subsection{Uniqueness of improving edges}\label{sec: move to vertex}
We now show that in every vertex of the hypercube that is not optimal, there is exactly one improving edge direction.

\begin{proposition}\label{proposition: exactly one}
For all~$n\in\N$ and all~$\vec x\in\{0,1\}^n\setminus\{\vec{e^n}\}$, there is exactly one~$k\in[n]$ such that~$\partial_k \fct n(\vec x)>0$ and~$x_k=0$, or~$\partial_k \fct n(\vec x)<0$ and~$x_k=1$.
\end{proposition}

Before giving the proof, we introduce notation tightening our further analysis.

\begin{definition}
    Given some~$\vec x=(x_1, \ldots, x_n)^{\top}\in\R^n$ and~$k\in[n]$, we write
    \[
    \pp\nolimits_n (\vec x,k)=x_{k-1}\prod\limits_{j=1}^{k-2}(1-x_j) \quad \text{and} \quad \s\nolimits_n(\vec x,k)=\sum\limits_{j=k+1}^{n} x_j,
    \]
    where we set~$x_0\coloneqq 1$, i.e., we have~$\pp_n(\vec x,1)=1$ for all~$\vec x$.
    Whenever~$n$ and~$\vec x$ are clear from the context, we write~$\pp(k)\coloneqq \pp\nolimits_n (\vec x,k)$ and~$\s(k)\coloneqq \s\nolimits_n(\vec x,k)$.
\end{definition}

We now observe how the dimension~$k$ of Proposition~\ref{proposition: exactly one} is determined by the values of~$\pp_n$ and~$\s_n$.

\begin{restatable}{lemma}{lemmaimprovableiff}\label{lemma: k-improvable iff}
The following are equivalent for all~$\vec{x}\in\{0,1\}^n$ and~$k\in[n]$:
\begin{enumerate}[itemsep=8pt]
\item $\partial_k \fct n(\vec x)>0$ and~$x_k=0$, or~$\partial_k \fct n(\vec x)<0$ and~$x_k=1$,
\item $\s(k)\equiv_2 x_k$ and~$\pp(k)=1$.
\end{enumerate}
\end{restatable}
\begin{proof}
Fix~$\vec{x}\in\{0,1\}^n$ and~$k\in[n]$.
We claim that~$\alpha_{k+1}\equiv_2 \s(k)$.
By induction on~$i$ and the definition of $\alpha_i$ in equation~\eqref{def alpha}, it is not hard to see that
\begin{align*}
	\alpha_i=\sum_{j=i}^n x_j\prod_{\ell=i}^{j-1}(1-2x_{\ell}) 
\end{align*}
holds for all~$i\in[n+1]$.
This yields
\begin{align*}
\alpha_{k+1}&{=}\sum_{j=k+1}^n x_j\prod_{\ell=k+1}^{j-1}(1-2x_{\ell}) 
= \sum_{j\in\Delta_0}x_j - \sum_{j\in\Delta_1}x_j 
= \begin{cases} 0, \quad \text{if $\s(k)\equiv_2 0$}, \\ 1, \quad \text{if $\s(k)\equiv_2 1$},\end{cases}
\end{align*}
where we write
\[
\Delta_z\coloneqq\{j\in\{k+1,\ldots,n\}\colon |\{\ell\in\{k+1,\ldots,j-1\}\colon x_{\ell}=1\}|\equiv_2 z\}
\]
for~$z\in\{0,1\}$.
This proves our claim that~$\alpha_{k+1}\equiv_2 \s(k)$.

For the first part of the equivalence, assume that~$\s(k)\equiv_2 x_k$ and~$\pp(k)=1$. 
By~$\pp(k)=1$, we have~$\sum_{i=1}^{k-2}x_i+1-x_{k-1}=0$, so Lemma~\ref{lemma: partial F} yields
\begin{equation}\label{eq: partial in some case}
\partial_k\fct n(\vec x) = (1-2\alpha_{k+1})\sum_{i=1}^k 2^{i-1}\prod_{j=i}^{k-1}(1-2x_j).
\end{equation}
By using~$\sum_{i=1}^{k-1} 2^{i-1}<2^{k-1}$ and~$\prod_{j=i}^{k-1}(1-2x_j)\in\{\pm 1\}$ for all~$i\in[n+1]$, this yields~$\partial_k\fct n(\vec x)>0$ if~$\alpha_{k+1}=0$, and~$\partial_k\fct n(\vec x)<0$ if~$\alpha_{k+1}=1$.
Observe that we have~$\alpha_i\in\{0,1\}$ for all~$i\in[n+1]$, e.g., by induction on~$i$ starting with~$i=n+1$, so~$x_k\equiv_2\s(k)\equiv_2\alpha_{k+1}$ implies~$x_k=\alpha_{k+1}$.

For the second part of the equivalence, assume that we have~$\partial_k \fct n(\vec x)>0$ and~$x_k=0$, or~$\partial_k \fct n(\vec x)<0$ and~$x_k=1$.
By Lemma~\ref{lemma: partial F} and~$\sum_{i=1}^k 2^{i-1}<2^k$, this yields~$\sum_{i=1}^{k-2}x_i+1-x_{k-1}=0$, which implies~$\pp(k)=1$.
Then, Lemma~\ref{lemma: partial F} and~$\sum_{i=1}^{k-1} 2^{i-1}<2^{k-1}$ yield~$\sign (1-2\alpha_{k+1}) = \sign \partial_k \fct n(\vec x)$, where~$\sign$ denotes the signum function\footnote{The signum function~$\sign\colon\R\to\N$ is given by~$\sign( x)=-1$ if~${x}<0$, ~$\sign( x)=0$ if~${x}=0$, and~$\sign( x)=1$ if~${x}>0$.}, so we obtain~$\alpha_{k+1}=0$ if~$\partial_k\fct n(\vec x)>0$, and~$\alpha_{k+1}=1$ if~$\partial_k\fct n(\vec x)<0$.
By assumption, this yields~$x_k=\alpha_{k+1}\equiv_2 \s(k)$.
\qed
\end{proof}

We obtain our proposition by showing that there is a unique dimension~$k$ that satisfies condition~$(ii)$ from Lemma~\ref{lemma: k-improvable iff} in each vertex that is not optimal.

~

\noindent\textit{Proof of Proposition~\ref{proposition: exactly one}.}
We assume~$n\geq 3$, the cases~$n\in\{1,2\}$ can be checked directly.
We start with two technical observations.
First, fix some~$\vec x\in\{0,1\}^n$ and~$1<k_1<k_2\leq n$.
We can easily see that~$\pp(k_1)=1$ implies~$x_{k_1-1}=1$, while~$\pp(k_2)=1$ implies~$x_{k_1-1}=0$.
Since~$\pp(k)\in\{0,1\}$ for all~$k\in[n]$, this yields the following:
\begin{equation}\label{eq: note 1}
\forall\vec x\in\{0,1\}^n,1<k_1<k_2\leq n\colon\pp(k_1)=0\lor\pp(k_2)=0.
\end{equation}

Now fix some~$\vec x\in\{0,1\}^n$ with~$\s(1)\equiv_2 x_1$.
Assume that~$\pp(k)=1$ holds for some fixed~$k\in\{2,\ldots,n\}$.
Then, we have~$x_{k-1}=1$ and~$x_j=0$ holds for all~$j\in[k-2]$.
Thus,~$\s(1)\equiv_2 x_1$ implies~$\s(k)\equiv_2 x_k+1$, which gives us the following:
\begin{equation}\label{eq: note 2}
\forall\vec x\in\{0,1\}^n \text{ with }\s(1)\equiv_2 x_1, ~\forall k\in\{2, \ldots, n\}\colon\pp(k)=0\lor\s(k)\equiv_2 x_k+1.
\end{equation}

For all~$\vec x\in\{0,1\}^n$, equations~\eqref{eq: note 1} \& \eqref{eq: note 2} yield that we have~$\s(k)\equiv_2 x_k$ and~$\pp(k)=1$ for {at most one}~$k\in[n]$.
	
Now consider some~$\vec x\in\{0,1\}^n\setminus\{\vec{e^n}\}$.
It remains to argue that there exists some~$k\in[n]$ with~$\s(k)\equiv_2 x_k$ and~$\pp(k)=1$.
We may assume~$\s(1)\equiv_2 x_1+1$ since~$\pp(1)=1$.
This yields the existence of some~$j\in[n]$ with~$x_j=1$. 
Further, we have~$k\coloneqq\min\{j\in[n]\colon x_j=1\}<n$ as~$\vec x\neq \vec{e^n}$, hence~$\pp(k+1)=1$ and~$\s(k+1)\equiv_2 x_{k+1}$.
The statement now follows directly by Lemma~\ref{lemma: k-improvable iff}.	
\qed

~

Let~$k$ be the unique dimension from Proposition~\ref{proposition: exactly one} for some vertex~$\vec x$ that is not optimal. 
We show that the value of~$\partial_k\fct n(\vec x)$ does not depend on~$x_k$.

\begin{restatable}{proposition}{propmovetovertex}\label{proposition: move to next vertex}
Let~$n\in\N$ and~$\vec x\in\{0,1\}^n\setminus\{\vec{e^n}\}$.
Let~$k\in[n]$ be the unique dimension from Proposition~\ref{proposition: exactly one}.
Then, we have
\[
\partial_k\fct n(\vec x)=\partial_k\fct n(\vec x+\mu\vec{e^k})
\]
for all~$\mu\in\R$.
\end{restatable}
\begin{proof}
By Lemma~\ref{lemma: k-improvable iff}, we have $\s(k)\equiv_2 x_k$ and~$\pp(k)=1$.
Hence, equation~\eqref{eq: partial in some case} holds in~$\vec x$.
In particular, the value of~$\partial_k\fct n(\vec x)$ does not depend on the value of~$x_k$.
\qed
\end{proof}

\subsection{Visiting all vertices of the hypercube}\label{sec: hamiltonian}
We will see that \textsc{ActiveSet} applied to~$\fct n$ over the hypercube and starting in~$\vec{0}$ moves from one vertex to another along the unique improving edges we identified in Section~\ref{sec: move to vertex}.
This trajectory results in a Hamiltonian path in the polyhedral graph of~$[0,1]^n$, i.e., in a path along the edges of the hypercube that visits all of its~$2^n$ vertices exactly once.

\begin{proposition}\label{proposition: Hamiltonian Path}
    Let~$n\in\N$.
    The sequence~$(\vec{x^i})_{i\in[2^n]}\subset\{0,1\}^n$ given by~$\vec{x^1}\coloneqq\vec 0$ and
    \begin{equation}\label{eq: hamiltonian sequence}
    \vec{x^{i+1}}\coloneqq\vec{x^i}+(1-2(x^i)_{k_i})\vec{e^{k_i}}, \quad \forall i\in[2^n-1],   
    \end{equation}
where~$k_i\in[n]$ is the unique dimension of Proposition~\ref{proposition: exactly one} in~$\vec{x^i}$, forms a Hamiltonian path in the polyhedral graph of~$[0,1]^n$ ending in~$\vec{e^n}$.
\end{proposition}
\begin{proof}
	By Lemma~\ref{lemma: k-improvable iff} and Proposition~\ref{proposition: exactly one}, the choice of~$k_i\in[n]$ is indeed unique for~$\vec{x^i}\in\{0,1\}^n\setminus\{\vec{e^n}\}$ and we have~$\pp(k_i)=1$ and~$\s(k_i)\equiv_2 (x^i)_{k_i}$.
	We prove the statement by induction on~$n\in\N$. 
	It obviously holds for~$n=1$, so assume it holds for some fixed~$n\in\N$.
	
	Consider the sequence~$(\vec{x^i})_{i\in[2^{n+1}]}\subset\{0,1\}^{n+1}$ as defined in the statement.
	By induction, the subsequence~$(\vec{x^i})_{i\in[2^{n}]}$ forms a Hamiltonian path in the polyhedral graph of~$[0,1]^n$ ending in~$\vec{x^{2^n}}=\vec{e^n}$, since the additional entry~$(x^i)_{n+1}=0$ does not change the value of~$\pp(k)$ or~$\s(k)$ for any~$k\in[n]$.
	Further, since we have~$\pp_{n+1}(\vec{e^{n}},n+1)=1$ and~$\s_{n+1}(\vec{e^n},n+1)=0= (\vec{e^n})_{n+1}$, it holds that~$k_{2^n}=n+1$ and thus~$\vec{x^{2^{n}+1}}=\vec{e^{n}}+\vec{e^{n+1}}$.
	
	For the remaining part of the sequence, note that having~$(x^i)_{n+1}=1$ instead of~$(x^i)_{n+1}=0$ in the additional dimension does not change~$\pp_{n+1}(\vec{x^i},k)$ and flips the value of~$(\s_{n+1}(\vec{x^i},k)\mod 2)\in\{0,1\}$ for all~$k\in[n]$.
	
	By equation~\eqref{eq: hamiltonian sequence}, for all~$i\in[2^{n}-1]$, we have~$\pp_{n+1}(\vec{x^i},k_i)=\pp_{n+1}(\vec{x^{i+1}},k_i)$ and~$\s_{n+1}(\vec{x^i},k_i)=\s_{n+1}(\vec{x^{i+1}},k_i)$.
	Therefore, the second half of the sequence is the inverse of the first half; more precisely, we have
	\[
	(\vec{x^i})_{i\in\{2^{n}+1,2^n+2,\ldots,2^{n+1}\}}= (\vec{y^k})_{k\in[2^{n}]},
	\] 
	where~$\vec{y^k}\coloneqq \vec{x^{2^{n}+1-k}}+\vec{e^{n+1}}$ for all~$k\in[2^{n}]$.
	In particular,~$(\vec{x^i})_{i\in[2^{n+1}]}$ forms a Hamiltonian path in the polyhedral graph of~$[0,1]^{n+1}$ ending in~$\vec{e^{n+1}}$.
	\qed
\end{proof}

\subsection{Proofs of the super-polynomial bounds}\label{sec: main results}

We define a matrix~$C\in\{-1,0,1\}^{2n\times n}$ by~$\vec{C_{i\cdot}}=(\vec{e^i})^{\top}$ and~$\vec{C_{(i+n)\cdot}}=-(\vec{e^i})^{\top}$ for all~$i\in[n]$, and a vector~$\vec c\in\{0,1\}^{2n}$ by~$c_i=1$ and~$c_{i+n}=0$ for all~$i\in[n]$.
So we have
\[
[0,1]^n=\{\vec x\colon C\vec x\leq \vec c\}.
\]
Fixing this implementation of the hypercube, we can now state our main result formally.
For the proof, we combine the insights of the previous subsections.

~

\noindent \textbf{Theorem~\ref{thm:main}.} 
\textit{
For all~$n\in\N$,~$\textsc{ActiveSet}(\fct n,C,\vec c,\vec{0})$ solves
\begin{align*}
    \max \quad & \fct n(\vec x) \\
    \text{s.t.} \quad & \vec x \in [0,1]^n
\end{align*}
in~$2^n-1$ iterations, irrespective of the pivot rule.}
\begin{proof}
We show that \textsc{ActiveSet} visits all vertices of the hypercube~$[0,1]^n$.
Let~$(\vec{x^i})_{i\in[2^n]}$ denote the sequence of Proposition~\ref{proposition: Hamiltonian Path}.
The initial solution is given by~$\vec 0=\vec{x^1}$.
Now fix some~$i\in[2^n-1]$ and assume that~$\vec{x^i}$ is the intermediate solution at the start of iteration~$i$ of $\textsc{ActiveSet}(\fct n,C,\vec c,\vec{0})$ and that the current active set is~$\mathcal{A}=\Eq(\vec{x^i})\coloneqq\{j\colon\vec{C_{j\cdot}}\vec{x^i}=c_j\}$.
Note that, according to Proposition~\ref{proposition: Hamiltonian Path}, we have~$\vec{x^i}\neq\vec{e^n}$ so we can fix~$k\in[n]$ as the unique dimension of Proposition~\ref{proposition: exactly one} in the point~$\vec{x^i}$.

All edge directions in~$\vec{x^i}$ are of the form
\[
\vec{d^{j,\lambda}}\coloneqq \lambda(1-2(x^i)_j)\vec{e^j}
\] 
for some~$j\in[n]$ and some~$\lambda>0$.
Note that~$\partial_j \fct n=\vec{\nabla\!\fct n}\!^{\top}\vec{e^j}=-\vec{\nabla\!\fct n}\!^{\top}(-\vec{e^j})$ for all~$j\in[n]$, so we have~$\vec{\nabla\!\fct n}(\vec{x^i})\!^{\top}\vec{d^{j,\lambda}}>0$ if and only if~$\partial_j \fct n(\vec{x^i})>0$ and~$(x^i)_j=0$, or~$\partial_j \fct n(\vec{x^i})<0$ and~$(x^i)_j=1$.
Hence, by Proposition~\ref{proposition: exactly one}, the improving edge directions in~$\vec{x^i}$ are given by~$\vec{d^{k,\lambda}}=\lambda(1-2(x^i)_k)\vec{e^k}$ for~$\lambda>0$.

Since the system~$C\vec x\leq \vec c$ is non-degenerate and~$\mathcal{A}=\Eq(\vec{x^i})$, every feasible direction~$\vec d$ in the vertex~$\vec{x^i}$ satisfies
\[
|\{i\in\mathcal{A}\colon \vec{A_{i\cdot}}\vec{ d}=0\}|=|\mathcal{A}|-1
\]
if and only if it is an edge direction in~$\vec{x^i}$.
In particular, for every edge direction~$\vec d$ in~$\vec{x^i}$, there is a unique index~$\ell\in A$ with~$A_{\ell\cdot}\vec d<0$ and~$A_{\mathcal{A}\setminus\{\ell\}\cdot}\vec d=\vec 0$.
Further, the system~$A_{\mathcal{A}\cdot}\vec d=\vec 0$ only has the trivial solution.
Hence, taking an arbitrary feasible improving direction maximizing~$|\{i\in\mathcal{A}\colon \vec{A_{i\cdot}}\vec{ d}=0\}|$ is equivalent to taking an arbitrary improving edge direction.

This yields that \textsc{ActiveSet} chooses the improving edge direction~$\vec d\coloneqq\vec{d^{k,\lambda}}$ for some fixed~$\lambda>0$ in iteration~$i$ and deletes the unique~$\ell\in\mathcal{A}$ with~$A_{\ell\cdot}\vec{d}<0$.
We have~$A_{\mathcal{A}\setminus\{\ell\}\cdot}\vec{d}=\vec 0$ and, by Proposition~\ref{proposition: move to next vertex},
\[
\vec{\nabla\!\fct n(x^i}+\mu \vec{d)}\!^{\top}\vec{d}=\vec{\nabla\!\fct n}(\vec{x^i})\!^{\top}\vec{d}> 0
\]
for all~$\mu\in\R$.
Hence, \textsc{ActiveSet} moves from~$\vec{x^i}$ along~$\vec{d}$ in iteration~$i$ until reaching the next vertex
\[
\vec{x^i}+\vec{d}=\vec{x^i}+(1-2(x^i)_k)\vec{e^k}=\vec{x^{i+1}}.
\]
Finally, since~$\vec{\nabla\!\fct n(x^{i+1})}\!^{\top}\vec{d}> 0$, an index~$j\in\Eq(\vec{x^{i+1}})\setminus\mathcal{A}$ is added to~$\mathcal{A}$.
In fact, since the system~$C\vec x\leq \vec c$ is non-degenerate, this choice of~$j$ is unique, i.e., \textsc{ActiveSet} starts iteration~$i+1$ in the point~$\vec{x^{i+1}}$ with~$\mathcal{A}=\Eq(\vec{x^{i+1}})$.

We have seen that \textsc{ActiveSet} transitions from~$\vec{x^i}$ to~$\vec{x^{i+1}}$ in iteration~$i$, where~$i\in[2^n-1]$ was chosen arbitrarily.
This proves the statement since the optimum solution of the given program is~$\vec{x^{2^n}}=\vec{e^n}$ by Lemma~\ref{lemma: maximum}.
\qed
\end{proof}

If we are willing to relax the exponential bound from Theorem~\ref{thm:main}, we can lower the degrees of the polynomials.
In Corollary~\ref{cor:main2}, the degrees are chosen just large enough to yield a super-polynomial bound.

~

\noindent\textit{Proof of Corollary~\ref{cor:main2}.}
Let~$d(n)=\omega(\log n)$ with~$d(n)\leq n$ for all~$n\in\N$.
Since the additional dimensions do not affect its behavior, by Theorem~\ref{thm:main}, \textsc{ActiveSet} started in~$\vec 0$ solves
\begin{align*}
    \max \quad & \fct {\lfloor d(n)\rfloor}(x_1,\dots,x_{\lfloor d(n)\rfloor}) \\
    \text{s.t.} \quad & \vec x \in [0,1]^n
\end{align*}
in~$2^{\lfloor d(n)\rfloor}-1$ iterations for all~$n\in\N_{>2}$.
Since~$2^{\lfloor d(n)\rfloor}=2^{\omega(\log n)}$ and the polynomials~$\left(\fct {\lfloor d(n)\rfloor}\right)_{n\in\N}$ are of degrees~${\lfloor d(n)\rfloor}=\mathcal{O}(d(n))$, this concludes the proof.
\qed

\subsection{The complexity of maximizing over the hypercube}\label{sec: complexity}

The lower bounds in Theorem~\ref{thm:main} and Corollary~\ref{cor:main2} are unconditional, meaning they show that \textsc{ActiveSet} is inefficient on the hypercube even if~$\p=\np$.
It is folklore that the underlying decision problem is already~\np-hard for polynomials of degree~$3$. 
We include a proof to make our presentation self-contained.

\begin{restatable}{proposition}{propnpcomplete}\label{prop: NP complete}
Given a multivariate polynomial~$f\colon\R^n\to\R$ of degree~$3$, it is~\np-complete to decide whether there is~$\vec x\in\{0,1\}^n$ with~$f(\vec x)\geq 0$.
\end{restatable}
\begin{proof}
We can evaluate polynomials of degree~$3$ efficiently, so the decision problem is in \np.
For hardness, we reduce from 3\textsc{SAT}~\cite{Cook1971}.
Let an instance for 3\textsc{SAT} be given by clauses~$\mathcal{C}=\{C_j\}_{j=1,\ldots,m}$ over variables~$\mathcal{Z}=\{z_i\}_{i=1,\ldots,n}$.
Define the function~$f\colon\R^n\to\R$ by
\[f(\vec x)=-\sum\limits_{j=1}^m \prod_{\{k\in[n]\colon z_k\in C_j\}}\prod_{\{\ell\in[n]\colon \neg{z_{\ell}}\in C_j\}}(1-x_k)x_{\ell}.\]
Then, for all~$\vec x\in\{0,1\}^n$, we have~$f(\vec { x})=0$  if and only if~$\sigma\colon\mathcal{Z}\to\{0,1\}$ given by~$\sigma(z_i)\coloneqq x_i$ is an assignment satisfying the 3\textsc{SAT} instance.
Since~$f(\vec x)\leq 0$ for all~$\vec { x}\in[0,1]^n$, we obtain \np-hardness of the decision problem.
\qed
\end{proof}

While the underlying problem is hard, the next statements show that our polynomial~$\fct n$ belongs to a class of functions for which linear time algorithms exist that find the optimum vertex of the hypercube.
We need some definitions before we can state this formally.

Every function~$f\colon \R^n\to\R$ attaining unique values on the vertices of~$[0,1]^n$ induces an {orientation}~$\sigma_f$ of the edges of the hypercube if we direct every edge~$\{\vec x,\vec y\}$ from~$\vec x$ to~$\vec y$ if~$f(\vec x)<f(\vec y)$, and from~$\vec y$ to~$\vec x$ if~$f(\vec y)<f(\vec x)$.
We write~$\vec{x}\to_{\sigma}\vec y$ if edge~$\{\vec x,\vec y\}$ is directed from~$\vec x$ to~$\vec y$ in~$\sigma$.

An orientation~$\sigma$ of the edges of the hypercube is called \emph{unique sink orientation} (USO) if there is exactly one \emph{sink}, i.e., a vertex without outgoing edges, in every face of the hypercube.
For more details on USOs, see e.g.~\cite{gartner2006linear,szabo2001unique}.

We say that an orientation~$\sigma$ is \emph{combed} if there is a dimension~$i\in[n]$ such that all edges along dimension~$i$ point in the same direction, i.e., if we have
\[
\vec x\to_{\sigma}\vec x+(1-2x_i)\vec{e^i}
\]
for all~$\vec{x}\in\{0,1\}^n$ with~$x_i=\delta$, where~$\delta\in\{0,1\}$ is fixed.
An orientation is called \emph{decomposable} if it is combed on every nonzero-dimensional subcube of~$[0,1]^n$.

\begin{restatable}{proposition}{propuso}\label{prop: induces USO}
The function~$\fct n$ induces a decomposable orientation of~$\{0,1\}^n$.
\end{restatable}
\begin{proof}
Let~$I=\{i_1,\ldots,i_k\}$ with~$1\leq i_1<\ldots<i_k\leq n$ denote the support of some subcube~$\mathcal{S}$ of~$\{0,1\}^n$, i.e., the subcube~$\mathcal{S}\subset\{0,1\}^n$ is $k$-dimensional and, for all~$i\in [n]\setminus I$, there is~$\delta_i\in\{0,1\}$ such that~$x_i=\delta_i$ for all~$\vec{x}\in\mathcal{S}$.
We will show that~$\mathcal{S}$ is combed in dimension~$i_k$ with respect to the orientation~$\sigma$ induced by~$\fct n$.
Let~$\mathcal{S}_0\coloneqq\{\vec x\in\mathcal{S}\colon x_{i_k}=0\}$ and~$\mathcal{S}_1\coloneqq\{\vec x\in\mathcal{S}\colon x_{i_k}=1\}$.

Since~$\beta_i=0$ on~$\{0,1\}^n$ for all~$i\in[n]$, we have~$\fct n(\vec x)=\sum_{i=1}^n 2^{i-1}\alpha_i(\vec x)$ for all~$\vec x\in\{0,1\}^n$.
Fix some arbitrary~$\vec x\in\mathcal{S}_0$ and~$\vec y\in\mathcal{S}_1$.
Then,~$x_j=y_j$ for all~$j>i_k$ and thus~$\alpha_i(\vec x)=\alpha_i(\vec y)$ for all~$i>i_k$.
Therefore, we have
\[
	\alpha_{i_k}(\vec x)=\alpha_{i_k+1}(\vec x)=\alpha_{i_k+1}(\vec y)=1-\alpha_{i_k}(\vec y).
\]
As~$\alpha_i\in\{0,1\}$ on~$\{0,1\}^n$ for all~$i\in[n+1]$ and~$\sum_{i=1}^{i_k-1}2^{i-1}<2^{i_k-1}$, this yields that we either have~$\fct n(\vec x)>\fct n(\vec y)$ for all~$\vec x\in\mathcal{S}_0$ and~$\vec y\in\mathcal{S}_1$ or~$\fct n(\vec x)<\fct n(\vec y)$ for all~$\vec x\in\mathcal{S}_0$ and~$\vec y\in\mathcal{S}_1$, i.e., the subcube~$\mathcal{S}$ is combed with respect to~$\sigma$.
\qed
\end{proof}

Every decomposable orientation is also a USO, and can be solved in linear time~\cite{SchurrS04}.
Hence, the polynomials of Theorem~\ref{thm:main} and Corollary~\ref{cor:main2} do not induce \np-hard optimization problems (unless~$\p = \np$).

\section{Future research}

In this paper, we considered the active-set method as a natural generalization of the simplex method in order to derive the first unconditional running time bounds for all pivot rules.
A canonical next step towards a general bound for the simplex method would be to lower the degree of our polynomial.

\begin{question}
Is there a polynomial of degree $\mathcal{O}(\log n)$ over the hypercube for which active-set takes super-polynomially many iterations for all pivot rules?
\end{question}

It might be possible to reduce the complexity of the objective function even more drastically by increasing the complexity of the feasible region.

\begin{question}
Is there a polynomial of constant degree over some polytope for which active-set takes super-polynomially many iterations for all pivot rules?
\end{question}

While a general bound for linear objectives, i.e., for the simplex method, would be a massive breakthrough, reaching quadratic and concave objectives would already be very interesting:
Much like in the linear setting, only weakly polynomial algorithms are known, namely ellipsoid and interior-point methods~\cite{kozlov1979polynomial,ye1989extension}, and the active-set method is a promising candidate for a strongly polynomial algorithm, due to its combinatorial nature.

\begin{question}
Is there a concave, quadratic polynomial over some polytope for which active-set takes super-polynomially many iterations for all pivot rules?
\end{question}

%
\bibliographystyle{splncs04}
\bibliography{DisserMosis2025}
\end{document}